\documentclass{article}
\usepackage{stmaryrd,amssymb,amsmath,amsthm,url}
\usepackage{latexsym} 
\usepackage{alltt}
\usepackage[all]{xy} 
\CompileMatrices 
\xyoption{frame}

\usepackage{amscd} 

\theoremstyle{plain}
\newtheorem{theorem}{Theorem}[section]
\newtheorem{corollary}[theorem]{Corollary}
\newtheorem{lemma}[theorem]{Lemma}
\newtheorem{proposition}[theorem]{Proposition}
\theoremstyle{definition}
\newtheorem{definition}[theorem]{Definition}
\newtheorem{notation}[theorem]{Notation}
\newtheorem{example}[theorem]{Example}
\newtheorem{examples}[theorem]{Examples}

\newcommand{\SA}{\axname{Signs}}
\newcommand{\Mod}{\axname{Mod}}

\newcommand{\sg}{\operatorname{{\mathbf s}}}
\newcommand{\di}{\mathsf{D}}
\newcommand{\st}{\mathsf{S}}

\newcommand{\diamarrow}[1]{\ensuremath{\mathbin{   %
                  \setlength{\unitlength}{1.6ex}
                  \begin{picture}(3,2)
                  \put(.09,.1){$\langle~#1~\rangle$}
                  \put(.65,-.55){\vector(-1,-1){2.5}}
                  \put(2.35,-.55){\vector(1,-1){2.5}}
                  \end{picture}
                  }}}
\newcommand{\diamleftarrow}[1]{\ensuremath{\mathbin{ %
                  \setlength{\unitlength}{1.6ex}
                  \begin{picture}(3,2)
                  \put(.09,.1){$\langle~#1~\rangle$}
                  \put(.65,-.55){\vector(-1,-1){2.5}}
                  \put(2.35,-.55){\line(1,-1){2.5}}
                  \end{picture}
                  }}}
\newcommand{\prefa}[1]{\ensuremath{\mathbin{        %
                  \setlength{\unitlength}{1.6ex}
                  \begin{picture}(3,2)(-.54,.1)
                  \put(-0.3,.2){$\lbrack~#1~\rbrack$}
                  \put(.95,-.5){\line(0,-1){2.4}}
                  \end{picture}
                  }}}
\newcommand{\prefarrow}[1]{\ensuremath{\mathbin{    %
                  \setlength{\unitlength}{1.6ex}
                  \begin{picture}(3,2)(-.54,.1)
                  \put(-0.3,.2){$\lbrack~#1~\rbrack$}
                  \put(.95,-.5){\vector(0,-1){2.4}}
                  \end{picture}
                  }}}

\newcommand{\semone}[1]{[\! [ #1 ]\!]}
\newcommand{\pmap}{\overset{p}{\longrightarrow}}
\newcommand{\ins}{{\mathit{Ins}}}
\newcommand{\var}{{\mathit{Var}}}
\newcommand{\ter}{{\mathit{Ter}}}
\newcommand{\Nat}{{\mathbb N}}
\newcommand{\tr}{{\mathtt{true}}}
\newcommand{\fa}{{\mathtt{false}}}

\newcommand{\cp}{\texttt{cp}}
\newcommand{\setz}{\texttt{set:0}}
\newcommand{\seto}{\texttt{set:1}}
\newcommand{\setai}{\texttt{set:ai}}
\newcommand{\setmi}{\texttt{set:mi}}
\newcommand{\seta}{\texttt{set:a}}
\newcommand{\setm}{\texttt{set:m}}
\newcommand{\sets}{\texttt{set:s}}
\newcommand{\test}{\texttt{test:0}}

\newcommand{\axname}[1]{\ensuremath{\textit{#1}}}
\newcommand{\IL}{\axname{IL}}

\newcommand{\Md}{\axname{Md}}

\newcommand{\Ril}{\axname{RIL}}

\title{Straight-line instruction sequence completeness for total calculation on cancellation meadows\thanks{Partially supported by the 
	  Dutch NWO project \emph{Thread Algebra for Strategic Interleaving}, 
	  project number 612.000.421.
	  }}

\author{
	Jan A.\ Bergstra
    \and
	Inge Bethke\\
\\
  {\small
	  Section Software Engineering,
	  Informatics Institute,
	  University of Amsterdam}\\
	{\small URL: \url{www.science.uva.nl/~{inge,janb}}
	}
}
\date{}

\begin{document}

\maketitle

\begin{abstract}
A combination of program algebra with the theory of meadows is designed leading to a theory of
computation in algebraic structures which use in addition to a zero test and copying instructions
the instruction set $\{x \Leftarrow 0, x \Leftarrow 1, x\Leftarrow -x, x\Leftarrow x^{-1}, x\Leftarrow x+y,
x\Leftarrow x\cdot y\}$. It is proven that total functions on cancellation meadows can be computed
by straight-line programs using at most 5 auxiliary variables. A similar result  
is obtained for signed meadows. 
\end{abstract}


\section{Introduction}

\emph{Program algebra} is an approach to the formal description of the semantics of programming
languages.  It is a framework that permits algebraic reasoning about programs and
has been investigated in various settings (see e.g.\ \cite{BL02,BM84,BM88a,BM88b,W99}).

The theory of fields is a very active area which is not only of great theoretical interest but has also found
applications both within mathematics|combina\-to\-rics and algorithm analysis|as well as in engineering sciences
and, in particular, in coding theory and
sequence design. Unfortunately, since fields are not
axiomatized by equations only, Birkhoff's Theorem fails, i.e. fields do not
constitute a variety: they are not closed under products, subalgebras, and homomorphic images.
In \cite{BT07}, the concept of \emph{meadows} was introduced, structures very similar to fields|the
considerable difference being that meadows enjoy a total multiplicative inversion and do form a variety.

The aim of this paper is to combine these two areas of research in order to create a theory of computation 
in algebraic structures which can be used to investigate questions of definability
and complexity. 

Many computations in applied mathematics can be formulated as computations on fields. In many cases 
such computations terminate on all inputs yielding total functions. Replacing fields by meadows, which 
simplify their equational logic, we investigate
properties of instruction sequences which compute total functions on all meadows. 
We shall prove that total functions on cancellation meadows|meadows which in addition satisfy the inverse law known form the theory
of fields|can be computed by \emph{straight-line programs} 
with a bound supply of auxiliary variables.
These kind of programs have been amply investigated
and simplification and equivalence problems for several classes of straight-line programs over 
varying instruction sets are known (see e.g.\ \cite{IM83,IL83}).

The paper is organized as follows. In the next section we recall the basics of program algebra, thread algebra
 and meadows. Here the notion of program algebra refers to the concept introduced in \cite{BL02} which
 focuses on instruction sequences.
 In Section 3 we introduce instruction sequences for functions on the rational numbers. The main theorem is proven
 in Section 4. 
 We prove that total functions on cancellation meadows can be represented by a normal form without
tests and jumps which uses at most 5 auxiliary variables. This result is extended to signed cancellation meadows|cancellation
meadows that presuppose an ordering of its domain|in Section 5.

\section{The basics of program algebra, thread algebra and meadows}
In this section we recall program algebra, thread algebra  and  meadows. 
\subsection{Program algebra}
The programm algebra PGA was introduced in \cite{BL02}.

Assume $A$ is a set of constants with typical elements 
$\mathtt{a,b,c,\ldots}$.
Instruction sequences are of the following form ($k\in\Nat$):
\[I ::= \mathtt{a}\mid +\mathtt{a}\mid -\mathtt{a}
\mid \#k\mid\;! \mid I;I\mid I^\omega.\]
The first five forms above are called  
\emph{primitive instructions}. These are

\begin{itemize}
\item
\emph{basic instructions} $\mathtt a$ which
prescribe behaviours that are considered
indivisible and executable in finite time, and which return upon execution a Boolean reply value,
\item 
\emph{test instructions} obtained from basic instructions by prefixing them with either a $+$ 
(positive test instruction) or a $-$ (negative test instruction) which
control subsequent 
execution via the reply of their execution,
\item
\emph{jump instructions} $\#k$ which
prescribe to jump $k$ 
instructions ahead|if possible; otherwise deadlock
occurs|and generate no observable
behavior, and
\item the
\emph{termination instruction} $!$ which prescribes successful
termination, an event that is taken to be observable.
\end{itemize}

\emph{Finite instruction sequences}  are obtained from primitive  instructions using \emph{concatenation}: if
$I$ and $J$ are finite instruction sequences, then so is 
\[I;J\] 
which is the instruction sequence that lists 
$J$'s primitive instructions right after those of $I$. 
A special subclass of the finite instruction sequences are the so-called 
\emph{straight-line instruction sequences} which are finite instruction sequences \emph{without} tests and jumps.

\emph{Periodic instruction sequences} are defined using the repetition operator:
if $I$ is an instruction sequence, then
\[I^\omega\]
is the instruction sequence that repeats $I$ forever, thus $I;I;I;\ldots$.

In PGA, different types of equality are discerned,
the most simple of which is \emph{single-pass congruence},
identifying sequences that execute identical 
instructions.
For finite instruction sequences, single-pass
congruence boils
down to the associativity of concatenation, and is 
axiomatized by 
\[
(X;Y);Z=X;(Y;Z).
\]
In the sequel we
leave out brackets in repeated concatenations. In the case of infinite instruction
 sequences, additional axioms are needed.
Define $X^1=X$ and for $n>0$, $X^{n+1}=X;X^n$. 
According to~\cite{BL02},
 single-pass congruence for arbitrary instruction sequences
is axiomatized by the axiom schemes
PGA1-PGA4 in Table~\ref{pgas}.

\begin{table}
\centering
\hrule
\begin{align*}
(X;Y);Z&=X;(Y;Z)&
\mathrm{(PGA1)}\\[1mm]
(X^n)^\omega&=X^\omega
&\mathrm{(PGA2)}\\[1mm]
X^\omega;Y&=X^\omega
&\mathrm{(PGA3)}\\[1mm]
(X;Y)^\omega&=X;(Y;X)^\omega&\mathrm{(PGA4)}\\[1mm]
\#n{+}1;u_{1};\ldots;u_{n};\# 0 &= 
  \#0;u_{1};\ldots;u_{n};\# 0& \mathrm{(PGA5)} \\[1mm]
\#n{+}1;u_{1};\ldots;u_{n};\# m &= 
  \#n{+}m{+}1;u_{1};\ldots;u_{n};\# m& \mathrm{(PGA6)} \\[1mm]
(\#k {+} n{+}1;u_{1};\ldots;u_{n})^{\omega} &=
  (\#k;u_{1};\ldots;u_{n})^{\omega}&\mathrm{(PGA7)} \\[1mm]
X = u_{1};\ldots;u_{n};(v_{1};\ldots;v_{m{+}1})^{\omega}&\rightarrow~
  \#n{+}m{+}k{+}2;X = \#n{+}k{+}1;X~~
  &\mathrm{(PGA8)}
\end{align*}
\hrule
\caption{PGA-axioms for single-pass congruence}
\label{pgas}
\end{table}

Using the axioms PGA1--PGA4 and thus preserving single-pass congruence,
each instruction sequence can be rewritten into one of the following forms:
\begin{itemize}
\item[] $Y$ not containing repetition, or
\item[] $Y;Z^\omega$ with $Y$ and $Z$ not containing repetition.
\end{itemize}
Any instruction sequence in one of the two above forms is said to be in
\emph{first canonical form}.

Instruction sequences in first canonical form can be converted into \emph{second
canonical form}: a first canonical form in which no chained jumps
occur, i.e., jumps to jump instructions (apart from $\#0$),
and in which each non-chaining jump into the repeating
part is minimized.
The associated congruence is called \emph{structural congruence}
and is axiomatized in Table~\ref{pgas}. Note that axiom~PGA8
is an equational axiom, the implication is only used to enhance
readability.

Two examples, of which the right-hand sides are in second canonical form: 
\begin{align*}
\#2;\mathtt a;(\#5;\mathtt b;+\mathtt c)^\omega&=_{sc}
\#4;\mathtt a;(\#2;\mathtt b;+\mathtt c)^\omega,\\
+\mathtt a; \#2;(+\mathtt b;\#2;-\mathtt c;\#2)^\omega&=_{sc}
+\mathtt a; \#0;(+\mathtt b;\#0;-\mathtt c;\#0)^\omega.
\end{align*}
For each instruction sequence there exists a structurally equivalent
second canonical form.

For more information on PGA we refer to~\cite{BL02,PZ06}. 
\subsection{Thread algebra}
\label{sec:2.3}
Thread algebra is the behavioural semantics for PGA and was
introduced in e.g.\  \cite{BB03,BL02}
under the name Polarized Process Algebra.

Finite threads are defined inductively by:
\begin{eqnarray*}
\st&-&\text{\emph{stop}, the termination thread,}\\
\di&-&\text{\emph{inaction} or \emph{deadlock}, the inactive thread},\\
T\unlhd\mathtt a \unrhd T'&-&
\text{the \emph{postconditional
composition} of $T$ and $T'$ for action $\mathtt a $,}\\
&&\text{where $T$ and $T'$ are finite threads and $\mathtt a\in A$.}
\end{eqnarray*}
The behavior of the thread
$T\unlhd\mathtt a \unrhd T'$ starts with the \emph{action} $\mathtt
a$ and continues as $T$ upon reply $\tr$ to $\mathtt a$, and as $T'$
upon reply $\fa$. Note that finite threads always end in $\st$ or
$\di$.
We use \emph{action prefix} $\mathtt a \circ T$ as an abbreviation for
$T\unlhd\mathtt a \unrhd T$ and take $\circ$ to bind strongest.

For every finite thread there exists a finite upper bound to the number of consecutive
actions it can perform. The \emph{approximation operator} $\pi_n$ gives the behaviour
up to depth $n$ and is defined by
\begin{enumerate}
\item $\pi_0(T)=\di$,
\item $\pi_{n+1}(\st)=\st$,
\item $\pi_{n+1}(\di)=\di$, and
\item $\pi_{n+1}(T\unlhd\mathtt a \unrhd T')=\pi_n(T)\unlhd\mathtt a \unrhd \pi_n(T')$
\end{enumerate}
for finite threads $T, T'$ and $n\in \Nat$. Infinite threads are obtained as \emph{projective
sequences} of finite threads of the form $(T_n)_{n\in \Nat}$ where for every $n\in \Nat$,
$\pi_n (T_{n+1})=T_n$.

Upon its execution, a basic or test
instruction yields the equally named action in a post
conditional composition.
Thread extraction on PGA, notation 
\[|X|\]
with $X$ an instruction sequence,
is defined by the thirteen equations in Table~\ref{thirteen}.
In particular, note that upon the execution of a positive
test instruction $+\mathtt a$, the reply $\tr$ to $\mathtt a$
prescribes to continue with
the next instruction and $\fa$ to skip the next instruction
and to continue with the instruction thereafter; if no such 
instruction is available, deadlock occurs.  For the execution
of a negative
test instruction $-\mathtt a$, subsequent execution
is prescribed by the complementary replies.
\begin{table}
\centering
\hrule
\begin{align*}
&&|!| &= \st&
&&|!;X| &= \st\\[1mm]
&&|\mathtt a|&=\mathtt a\circ\di&
&&|\mathtt a;X| &=\mathtt a \circ |X|\\[1mm]
&&|\mathtt{{+}a}|&=\mathtt a\circ\di
\hspace{2.5cm}&
&&|{+}\mathtt{a};X| &= |X|\unlhd \mathtt a \unrhd |\#2;X|\\[1mm]
&&|{-}\mathtt a|&=\mathtt a\circ \di,&
&&|{-}\mathtt a;X| &= |\#2;X|\unlhd \mathtt a \unrhd |X|\\[4mm]
&&~~|\#k|&=\di&
&&|\#0;X| &= \di\\[1mm]
&&&&&&|\#1;X| &= |X|\\[1mm]
&&&&&&|\#k{+}2;u| &= \di\\[1mm]
&&&&&&~~|\#k{+}2;u;X| &= |\#k{+}1;X|
\end{align*}
\hrule
\caption{Equations for thread extraction,
where $\mathtt a$ ranges over the basic instructions, and $u$ over the
primitive instructions ($k\in\Nat$)}
\label{thirteen}
\end{table}

For an instruction sequence  in second canonical form, these equations
either yield a finite thread, or a so-called \emph{regular} thread, i.e.,
a finite state thread in which infinite paths can occur.
Each regular thread can be specified (defined) by a
finite number of recursive equations.
As an example, the regular thread $T$ specified by
\begin{align*}
T&=\mathtt a\circ T'\\
T'&=\mathtt c\circ T'\unlhd \mathtt b\unrhd(\st\unlhd \mathtt d\unrhd T)
\end{align*}
can be defined by 
$\mathtt{|a;(+b;\#2;\#3;c;\#4;+d;!;a)^\omega|}$.
A picture of this thread is

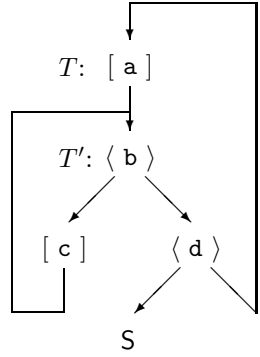
\setlength{\unitlength}{1.6ex}
\begin{picture}(30,20)(-20,0)
\put(1.6,10){$T'$:}
\put(1.6,15){$T$:}
\put(4,10){\diamarrow{\mathtt{b}}}
\put(4,15){\prefarrow{\mathtt{a}}}
\put(7.6,5){\diamleftarrow{\mathtt{d}}}
\put(0.4,5){\prefa{\mathtt{c}}}
\put(5,0){$\st$}
\put(12.45,1.95){\line(0,1){17.05}}
\put(5.5,19){\line(1,0){6.97}}
\put(5.5,19){\vector(0,-1){2}}
\put(-1,2){\line(1,0){2.9}}
\put(-1,2){\line(0,1){11}}
\put(-1,13){\line(1,0){6.5}}
\end{picture}
\\[4mm]
This thread can also given by the projective sequence $(\pi_n(T))_{n\in \Nat}$ where
\[
\begin{array}{rcl}
\pi_0(T)&=& \di\\
\pi_1(T)&=&  \mathtt a\circ\di\\
\pi_2(T)&=&\mathtt a\circ b \circ \di\\
\pi_3(T)&=&\mathtt a\circ (\mathtt c\circ \di \unlhd \mathtt b \unrhd \mathtt d \circ \di )
\end{array}
\]
and $\pi_{n+4}(T)= \mathtt a\circ (\mathtt c\circ \pi_{n+1}(T') \unlhd \mathtt b \unrhd ( \st \unlhd \mathtt b \unrhd \pi_{n+1}(T)))$.
Observe that thread extraction of straight-line instruction sequences yield finite and test-free threads.

For basic information on thread algebra we refer 
to~\cite{BBP05,PZ06}; more advanced matters, such as an operational
semantics for thread algebra, are discussed
in~\cite{BM07}. We here only mention the fact that 
each regular thread can be specified in PGA, and, conversely, 
that each PGA-program defines a regular thread.

\subsection{Meadows}
A {meadow}~\cite{BHT09,BT07} is a commutative ring with unit equipped with
a total unary operation $(\_)^{-1}$
named \emph{inverse} that satisfies the two equations
\begin{align*}
  (x^{-1})^{-1} &= x,   \\
  x\cdot(x \cdot x^{-1}) &= x. \quad(\Ril)
\end{align*}
Here \Ril\ abbreviates 
\emph{Restricted Inverse Law}.  We write \Md\ for the set of
axioms in Table~\ref{Md}.

\begin{table}
\centering
\rule[-2mm]{7cm}{.5pt}
\begin{align*}
	(x+y)+z &= x + (y + z)\\
	x+y     &= y+x\\
	x+0     &= x\\
	x+(-x)  &= 0\\
	(x \cdot y) \cdot  z &= x \cdot  (y \cdot  z)\\
	x \cdot  y &= y \cdot  x\\
	1\cdot x &= x \\
	x \cdot  (y + z) &= x \cdot  y + x \cdot  z\\
	(x^{-1})^{-1} &= x \\
	x \cdot (x \cdot x^{-1}) &= x
\end{align*}
\rule[3mm]{7cm}{.5pt}
\vspace{-5mm}
\caption{The set \Md\ of axioms for meadows}
\label{Md}
\end{table}

In the meadow $\mathbb{Q}$ of rational numbers, every element has a
restricted inverse. If $x\neq 0$, the inverse is just  the ``regular" inverse,
and $0^{-1}=0$. 
Another example is
ring $\mathbb{Z}/6\mathbb{Z}$ with elements $\{0,1,2, \ldots , 5\}$ where arithmetic is performed modulo $6$.
We find that every element has a restricted inverse as follows:
\[
\begin{array}{rclccrcl}
(0)^{-1} &=& 0&\hspace{1 cm}&
(1)^{-1} &=& 1\\
(2)^{-1} &=& 2&&
(3)^{-1} &=& 3\\
(4)^{-1} &=& 4&&
(5)^{-1} &=& 5\\
\end{array}
\]
A characterization of finite meadows can be found in \cite{BRS09}.
From the axioms in \Md\ the following identities are derivable:
\begin{align*}
	(0)^{-1}  &= 0,\\
	(-x)^{-1} &= -(x^{-1}),\\
	(x \cdot  y)^{-1} &= x^{-1} \cdot  y^{-1},
\\
0\cdot x  &= 0,\\
	x\cdot -y &= -(x\cdot y),\\
	-(-x)     &= x.
\end{align*}

We write $\Sigma_m=(0,1,+,\cdot,-,^{-1})$ for the
signature of meadows and $\ter(\Sigma_m,X)$ for the set of open meadow terms with free variables 
in $X$. For $t,u \in \ter(\Sigma_m,X)$ we shall often write $1/t$ 
or 
\[\frac 1 t \]
for $t^{-1}$, $tu$ for $t\cdot u$, $t/u$ for $t\cdot 1/u$, 
$t-u$ for $t+(-u)$, 
and freely use numerals $n$|abbreviating $\underbrace{1+ \cdots +1}_{n\times}$|and exponentiation with 
integer exponents as in $t^k$. 
We shall further write
\[1_x\text{ for } \frac x x\qquad\text{and}\qquad0_x\text{ for }1-1_x,
\]
so, $0_0=1_1=1$, $0_1=1_0=0$, and for all terms $t$,
\[0_t+1_t=1.\]

We write $\Sigma_r=(0,1,+,\cdot,-,)$ for the
signature of rings.
A \emph{polynomial} is an expression over $\Sigma_r$, thus  without 
inverse operator. Note that every polynomial can be represented as a sum of \emph{monomials}, 
i.e.\ products
of variables with integer coefficients. Meadow terms enjoy a particular standard
representation which was introduced in \cite{BP08}. 

\begin{definition}
A term $t\in \ter(\Sigma_m,X)$ is a \emph{Standard Meadow Form
(SMF)} if, for some $n\in\Nat$, $t$ is an \emph{SMF of level $n$}.
SMFs of level $n$ are defined as follows:
\begin{enumerate}
\item \textit{SMF of level $0:$} each expression of the form $s/t$ 
with $s$ and $t$ ranging over polynomials,
\item \textit{SMF of level $n+1:$} each expression of the form
\[0_{t'}\cdot s+1_{t'}\cdot t
\]
with $t'$ ranging over {polynomials} and $s$ and $t$ over SMFs of 
level $n$.
\end{enumerate}
\end{definition}
\begin{theorem}
\label{SMF}
For each  $t\in \ter(\Sigma_m,X)$ there exist  an SMF $t_{SMF}$
with the same variables such that $\Md \vdash t=t_{SMF}$.
\end{theorem}
\begin{proof}
See~\cite{BP08}.
\end{proof}
It follows that every meadow term is provably equal to a sum of quotients of polynomials.
\begin{corollary}\label{smf}
For every  $t\in \ter(\Sigma_m,X)$ there exist polynomials $s_0, t_0, \ldots , s_n, t_n$ such that
\[
\Md\vdash t=\frac{s_0}{t_0} + \ldots + \frac{s_n}{t_n}
\]
\end{corollary}
\begin{proof}
Let $t_{SMF}$ be a SMF of $t$. We employ induction on its level $n$. If $n=0$ then $t_{SMF}=s_0/t_0$ with $s_0$ and $t_0$
polynomials. Assume $n=m+1$. Then $t_{SMF}=0_{t''}\cdot s + 1_{t''}\cdot t'$ where $t''$ is a polynomial,
and $s,t'$ are SMF's of level m. By the induction hypothesis $s=s_0/t_0 + \ldots + s_k/t_k$ and 
$t'=u_0/v_0 + \ldots + u_l/v_l$ with $s_0, t_0, \ldots , s_k, t_k, u_1, v_1, \ldots , u_l, v_l$
polynomials. Then
\[
\begin{array}{rcl}
t_{SMF}&=&0_{t''}\cdot s + 1_{t''}\cdot t'\\[2mm]
&=& (1-\frac{t''}{t''})\cdot s + \frac{t''}{t''}\cdot t'\\[2mm]
&=& s - (\frac{t''s_0}{t''t_0} + \ldots + \frac{t''s_k}{t''t_k}) + 
\frac{t''u_0}{t''v_0} + \ldots + \frac{t''u_l}{t''v_l}\\[2mm]
&=& s +\frac{-t''s_0}{t''t_0} + \ldots + \frac{-t''s_k}{t''t_k} + 
\frac{t''u_0}{t''v_0} + \ldots + \frac{t''u_l}{t''v_l}
\end{array}
\]
and the last term is again a sum of quotients of polynomials.
\end{proof}
The term \emph{cancellation meadow} was introduced in~\cite{BPvdZ} 
for a \emph{zero-totalized field}|a field in which $0^{-1}=0$. Cancellation meadows satisfy in addition the so-called 
\emph{cancellation axiom}
\[
x \neq 0 ~\&~ x\cdot y = x\cdot z ~\longrightarrow~ y=z.\]
An equivalent version of the cancellation axiom  is the
\emph{Inverse Law} (\IL), i.e., the conditional axiom
\begin{align*}
  x\neq 0 ~\longrightarrow~ x\cdot x^{-1}=1.
 \quad(\IL)
\end{align*}
So \IL\ states that there are no proper zero divisors. 
(Another equivalent formulation of the cancellation property is 
$x\cdot y=0~\longrightarrow~x=0\text{ or }y=0$.) The rationals $\mathbb{Q}$ form a cancellation meadow,
$\mathbb{Z}/6\mathbb{Z}$ does not.

\section{Calculation on cancellation meadows}
Instruction sequences for functions on the rational numbers are designed in such a way that 
computations can be performed only with the aid of auxiliary
variables to which initially the input values are copied, and from which the final values are copied to the output.
\begin{definition}
\begin{enumerate}
\item 
We distinguish two infinite, countable sets of input and auxiliary variables
$\var_{in}=\{x_i\mid i \in \Nat \}$ and 
$\var_{aux}=\{a_i\mid i \in \Nat \}$, and a single output variable $y$.
$\var$ denotes the union of these variables. 
\item The instruction set $\ins(\mathbb{Q})$|instructions on the rational numbers|consists of the following input, 
auxiliary and output instructions:
\[
\begin{array}{rcl}
\ins(\mathbb{Q})_{in}&=&\{a.\cp(x) \mid a\in \var_{aux} \ \& \ x\in \var_{in}\},\\
\\
\ins(\mathbb{Q})_{aux} &=& \{a.\setz, a.\seto, a.\setai, a.\setmi\mid a \in \var_{aux}\},\\
&&\cup\{a.\seta(a'), a.\setm(a'),  \mid a,a' \in \var_{aux}\}\\
&&\cup \{a.\test \mid a \in \var_{aux}\}\\
\\
\ins(\mathbb{Q})_{out}& =& \{ y.\cp(a) \mid a\in \var_{aux}\}.
\end{array}
\]
\end{enumerate}
\end{definition}
Here $\texttt{ai}$ and $\texttt{mi}$ refer to the unary meadow operations of additive and multiplicative inversion, 
and $\texttt{a}$ and $\texttt{m}$ to  binary addition and multiplication. 
The intended meaning of these instructions is depicted in Table \ref{ins}. Since assignment instructions always succeed,
it is assumed that the returned
truth value  is $\tr$. An instruction of the form $a.\test$ is not an assignment instruction but a zero test and returns
a truth value depending on the value of $a$. 
\begin{table}[htbp]
\[
\begin{array}{crccrrclcc}
\hline\\
\hspace{1 cm}&a.\cp(x) & \hspace{1 cm} & [&a&\Leftarrow& x&]&\hspace{1 cm}\\
&a.\setz & & [&a&\Leftarrow& 0&]\\
&a.\seto &  & [&a&\Leftarrow &1&]\\
&a.\setai &  & [&a&\Leftarrow& -a&]\\
&a.\setmi &  & [&a&\Leftarrow &a^{-1}&]\\
&a.\seta (a') &  & [&a&\Leftarrow& a + a'&]\\
&a.\setm (a') &  & [&a&\Leftarrow& a\cdot a'&]\\
&a.\test && \\
&y.\cp (a) &  & [&y&\Leftarrow &a&]\\
\\
\hline
\end{array}
\]
\caption{The instruction set and its informal semantics}\label{ins}
\end{table}
\begin{examples}\label{insseq}
\begin{enumerate}
\item Consider the following straight-line instruction sequence $I_1$:
\[
\begin{array}{l}
a_0.\cp (x_0); a_1.\seto; a_1.\seta (a_1); a_0.\seta (a_1); \\
a_1.\cp (x_0); a_0. \setm (a_1); 
a_0.\setmi; y.\cp (a_0); !
\end{array}
\]
$I_1$ represents the total meadow mapping $x \mapsto ((x+ 2)x)^{-1}$: first the auxiliary 
variable $a_0$ is assigned the value of the input variable $x_0$ and then is raised by $2$, after which $a_0$ is multiplied by $x_0$, inverted and copied to the output variable
$y$.
\item The periodic instruction sequence $I_2$
\[
\begin{array}{l}
a_0.\cp (x_0); a_1.\cp (x_1); a_2.\seto; a_3.\seto; a_3.\setai;\\
(-a_1.\test; \#3;y.\cp (a_2); !; a_2.\setm (a_0); a_1.\seta (a_3))^\omega
\end{array}
\]
represents the partial mapping $(x_0,x_1)\mapsto x_0^{x_1}$: first, the two arguments are copied to the auxiliary variables $a_0$
and $a_1$, and $a_2$ and $a_3$ are assigned the constants $1$ and $-1$, respectively. In the repetition,
$a_2$ is multiplied by the first argument, and 
the second argument is decreased by 1 until the zero test succeeds and the value of $a_2$ is 
copied to the output. This partial meadow mapping is defined for all pairs of the form $\langle x, n
\rangle$.
\end{enumerate}
\end{examples}

Meadows are standard mathematical structures, and as such, they may be described using standard logical formalisms. Here, 
we shall use the following first-order predicate logic over meadows and regular threads consisting of
\begin{enumerate}
\item the constants $0$ and $1$,
\item countably infinite constants $c_0, c_1 , \ldots $,
\item the unary function symbols $-$ and $\ ^{-1}$, representing additive and multiplicative inversion,
\item the binary function symbols $+$ and $\cdot$, written infix and representing addition and multiplication,
\item for every regular thread $T$ and $k,n\in \Nat$, a $k+1$-ary termination predicate $R_{T,k,n}(\vec{x})$, describing the property
\emph{$T$ terminates on input $x_0, x_1, \ldots, x_k$ after at most $n$ steps},
\item the usual Boolean connectives and first-order quantifiers with variables 
ranging over elements of meadows.
\end{enumerate}
The standard interpretation of the termination predicates is given below. Since assignment instructions  always succeed, we may assume that 
regular threads corresponding to 
instruction sequences on rational numbers are of the form $\st$, $\di$, 
$T \trianglelefteq a_i.\test   \trianglerighteq T'$ or $\mathit{ins} \circ T$ where $ins$ is an assignment instruction. 
\begin{definition}
Let $\mathcal{M}$ be a meadow. 
\begin{enumerate}
\item If $\alpha$ is an assignment in $\mathcal{M}$, i.e.\ $\alpha \in \mathcal{M}^{\var}$,
$v\in \var$, and $m\in \mathcal{M}$, we denote by $\alpha [v:=m]$ the assignment
$\alpha'$ with 
\[
\alpha'(v') =
\begin{cases}
m & \text{ if } v'\equiv v\\
\alpha(v') & \text {otherwise.}
\end{cases}
\]
\item Let $T$ be a regular thread and $n\in \Nat$. $R_{T,n}^\mathcal{M}\subseteq\mathcal{M}^\var$ is defined inductively as 
follows.
\begin{enumerate}
\item 
\[
R_{T,0}^\mathcal{M}=
\begin{cases}
\mathcal{M}^ \var&\text{ if } T=S,\\
\emptyset& \text{ otherwise.}
\end{cases}
\]
\item 
\[
\begin{array}{rcl}
R_{S,n+1}^\mathcal{M}&=&\mathcal{M}^ \var\\[1mm]
R_{D,n+1}^\mathcal{M}&=&\emptyset\\[1mm]
R_{a_i.\cp (x_j)\circ T,n+1}^\mathcal{M}&=&  \{\alpha \in \mathcal{M}^{\var}\mid\alpha[a_i:=\alpha(x_j)] \in R_{T,n}^\mathcal{M}\}\\[1mm]
R_{a_i.\setz\circ T,n+1}^\mathcal{M}&=&  \{\alpha \in \mathcal{M}^{\var}\mid \alpha[a_{i}:=0]\in
R_{T,n}^\mathcal{M} \}\\[1mm]
R_{a_i.\seto\circ T,n+1}^\mathcal{M}&=&  \{\alpha \in \mathcal{M}^{\var}\mid \alpha[a_i:=1] \in R_{T,n}^\mathcal{M} \}\\[1mm]
R_{a_i.\setai\circ T,n+1}^\mathcal{M}&=&  \{\alpha \in \mathcal{M}^{\var}\mid \alpha[a_i:=-\alpha(a_i)] \in R_{T,n}^\mathcal{M}\}\\[1mm]
R_{a_i.\setmi \circ T,n+1}^\mathcal{M} &=&  \{\alpha \in \mathcal{M}^{\var}\mid \alpha [a_i:=\alpha(a_i)^{-1}]\in R_{T,n}^\mathcal{M}\}\\[1mm]
R_{a_i.\seta (a_j) \circ T, n+1}^\mathcal{M}&=&\{\alpha \in \mathcal{M}^{\var}\mid \alpha [a_i:=\alpha(a_i)+ \alpha(a_j)]\in R_{T,n}^\mathcal{M}\}\\[1mm]
R_{a_i.\setm (a_j) \circ T, n+1}^\mathcal{M} &=  &\{\alpha \in \mathcal{M}^{\var}\mid\alpha [a_i:=\alpha(a_i)\cdot \alpha(a_j)]\in R_{T,n}^\mathcal{M}\}\\[1mm]
R_{T\trianglelefteq a_i.\test   \trianglerighteq T', n+1}^\mathcal{M} &=&  
\{ \alpha\in R_{T,n}^\mathcal{M}\mid \alpha(a_i)=0\}\cup 
\{\alpha\in R_{T',n}^\mathcal{M} \mid \alpha(a_i)\neq 0\} \\
R_{y.\cp (a_j)\circ T, n+1}^\mathcal{M} &= &  \{\alpha \in \mathcal{M}^{\var}\mid\alpha [y:=\alpha(a_j)]\in R_{T,n}^\mathcal{M} \}
\end{array}
\]  
\end{enumerate}
\item For $k,n\in \Nat$ and regular thread $T$ we define $\semone{R_{T,k,n}}_\mathcal{M} \subseteq \mathcal{M}^{k+1}$ by
\[
\semone{R_{T,k,n}}_\mathcal{M}=\{ \langle \alpha(x_0), \ldots ,\alpha(x_k)\rangle
\mid \alpha \in R_{T,n}^\mathcal{M}\}.
\]
\end{enumerate}
\end{definition}
\begin{example}
We consider once more the instruction sequences $I_1$ and $I_2$ introduced in Example \ref{insseq}. 
\begin{enumerate}
\item It is easy to see that, if $I=ins_1; \ldots ; ins_n;!$ is a straight-line instruction sequence consisting
of $n$ assignment instructions and ending in a single final termination instruction, then $\mathcal{M}\models
\forall x_0, \ldots , x_k\ R_{|I|,k,n}(x_0, \ldots, x_k)$ for every meadow $\mathcal{M}$. Hence, in particular,
$\mathcal{M}\models \forall x \ R_{|I_1|,0,8}(x)$.
\item $|I_2|$ starts with 5 initializing actions and repeats 3 consecutive actions until the zero test succeeds
in which case termination occurs after a final copying action. We thus have for all meadows $\mathcal{M}$, 
$\mathcal{M}\models \forall x\ R_{|I_2|, 1, 3n+7}(x,n)$.
\end{enumerate}
\end{example}
\begin{lemma}\label{pred}
 For all $k,n\in \Nat$, regular threads $T$ and meadows $\mathcal{M}$,
 \begin{enumerate}
\item $R^\mathcal{M}_{T,n}\subseteq R^\mathcal{M}_{T,n+1}$
\item $R^\mathcal{M}_{T,n}=R^\mathcal{M}_{\pi_{n+1}(T),n}$
\item $\mathcal{M}\models \forall x_0, \ldots, x_k \ (R_{T,k,n}(x_0, \ldots, x_k) \longrightarrow R_{T,k,n+1}(x_0, \ldots, x_k))$
\item $\mathcal{M}\models \forall x_0, \ldots, x_k \ (R_{T,k,n}(x_0, \ldots, x_k) \longleftrightarrow R_{\pi_{n+1}(T),k,n}(x_0, \ldots, x_k))$
\end{enumerate}
\end{lemma}
\begin{proof} 
(1) and (2) are proven by  straightforward induction; (3) and (4) follow from (1) and (2), respectively. \end{proof}

The \emph{apply operator} has been introduced in~\cite{BP02} as a means to transform a given state machine according to a thread. Given a meadow, we view its assignments as state machines which can be transformed. The corresponding apply operator is then defined as follows.
\begin{definition}
Let $\mathcal{M}$ be a meadow
and $T$ be a finite thread. 
We define the  apply operator $T\bullet: \mathcal{M}^{\var}\cup \{\di\} \rightarrow \mathcal{M}^{\var}\cup \{\di\}$  as follows.
\[
\begin{array}{rcll}
T\bullet \di&=& \di\\
\st\bullet \alpha  &=&\alpha\\
\di\bullet \alpha &=& \di\\
(a_i.\cp (x_j)\circ T) \bullet \alpha &=&  T \bullet \alpha [a_i:=\alpha(x_j)]\\
(a_i.\setz\circ T) \bullet \alpha  &=&  T \bullet\alpha [a_i:=0]\\
(a_i.\seto\circ T) \bullet \alpha &=& T \bullet \alpha [a_i:=1]\\
(a_i.\setai  \circ T) \bullet \alpha &=&  T \bullet\alpha [a_i:=-\alpha(a_i)]\\
(a_i.\setmi\circ T)\bullet \alpha &=&  T \bullet \alpha [a_i:=\alpha(a_i)^{-1}]\\ 
(a_i.\seta (a_j)\circ T) \bullet \alpha  &=&  T \bullet\alpha [a_i:=\alpha(a_i)+ \alpha(a_j)]\\
(a_i.\setm (a_j)\circ T) \bullet \alpha &=& T \bullet \alpha [a_i:=\alpha(a_i)\cdot \alpha(a_j)]\\
( T\trianglelefteq a_i.\test   \trianglerighteq T')\bullet \alpha &=&  \begin{cases}
T\bullet \alpha & \text{if } \alpha(a_i)=0\\
T'\bullet \alpha & \text {otherwise}
\end{cases}\\
(y.\cp (a_j)\circ T) \bullet \alpha &=&  T \bullet\alpha [y:=\alpha(a_j)]
\end{array}
\] 
\end{definition}
For infinite threads $T$, the apply operator is defined on certain inputs if  $T$ terminates after finitely many steps.
\begin{lemma}
Let $\mathcal{M}$ be a meadow
and $T$ be a regular thread. Then for all $n\in \Nat$ and $\alpha\in R^\mathcal{M}_{T,n}$,
\begin{enumerate}
\item $\pi_{n+1}(T)\bullet \alpha\neq \di$, and
\item $\forall k>n\ \pi_{k}(T)\bullet \alpha =
\pi_{n+1}(T)\bullet \alpha$.
\end{enumerate}
\end{lemma}
\begin{proof}
By straightforward induction.
\end{proof}
We can therefore define partial mappings corresponding to regular threads as below.
\begin{definition}
Let $\mathcal{M}$ be a meadow.
\begin{enumerate}
\item $\alpha \in \mathcal{M}^{\var}$ is called \emph{initial} if $\alpha (v)=0$ for all
$v\in \var - \var_{in}$.
\item Let $T$ be a regular thread and $k\in \Nat$. Then
$\semone{T}_\mathcal{M}^{k}: \mathcal{M}^{k+1}\pmap \mathcal{M}$ denotes the partial mapping defined as follows:
\[
\semone{T}_\mathcal{M}^{k}(m_0, \ldots , m_k)=
\begin{cases}
(\pi_{n+1}(T)\bullet \alpha_{m_0,\dots , m_k})(y) & \text{ if $\alpha_{m_0,\dots , m_k}
\in R^\mathcal{M}_{T, n}$,}\\
\text{undefined} & \text{ if for all $n\in \Nat$, $\alpha_{m_0,\dots , m_k} \not\in R^\mathcal{M}_{T, n}$.}
\end{cases}
\]
where $\alpha_{m_0,\dots , m_k}\in \mathcal{M}^{\var}$ is the \emph{initial} assignment with $\alpha(x_i)=m_i$ for $0\leq i \leq k$ and 
$\alpha (v) = 0$ for $v\in \var - \{x_0, \ldots , x_k\}$. 
\end{enumerate}
\end{definition}
\begin{notation}
If $I$ is an instruction sequence we shall write
$\semone{I}_\mathcal{M}^{k}$ for the corresponding meadow mapping instead of $\semone{|I|}_\mathcal{M}^{k}$. 
Moreover, when dealing with partial mappings, we let the symbols $\uparrow$ and $\downarrow$
denote un- and definedness, respectively.
\end{notation}
\begin{example}\label{insfunc}
We consider again the instruction sequences $I_1$ and $I_2$ given in Example \ref{insseq}.
\begin{enumerate}
\item Observe that 
\[
\begin{array}{rcl}
|I_1|&=&a_0.\cp (x_0)\circ a_1.\seto\circ a_1.\seta (a_1)\circ a_0.\seta (a_1) \\
&&\circ a_1.\cp (x_0)\circ a_0. \setm (a_1)\circ 
a_0.\setmi\circ y.\cp (a_0)\circ \st.
\end{array}
\]
Thus
\[
\begin{array}{rcl} 
|I_1| \bullet \alpha &= &\alpha[a_0:=\alpha(x_0)][a_1:=1][a_1:= 2][a_0:= \alpha(x_0)+2] \\
&&[a_1:=\alpha(x_0)][a_0:=(\alpha(x_0)+2)\cdot \alpha(x_0)]\\
&&[a_0:= ((\alpha(x_0)+2)\cdot \alpha(x_0))^{-1}][y:=((\alpha(x_0)+2)\cdot \alpha(x_0))^{-1}]
\end{array}
\]
for every meadow $\mathcal{M}$ and every assignment $\alpha\in \mathcal{M}^\var$. Hence $\semone{I_1}_\mathcal{M}^{0}(m)=((m+2)m)^{-1}$.
\item The periodic thread $|I_2|$ satisfies the equations
\[
\begin{array}{rcl}
|I_2|&=&a_0.\cp (x_0)\circ a_1.\cp (x_1)\circ a_2.\seto\circ a_3.\seto\circ a_3.\setai\circ T\\
T&=& (y.\cp (a_2)\circ S)\trianglelefteq a_1.\test   \trianglerighteq  (a_2.\setm (a_0)\circ  
a_1.\seta (a_3)\circ T).
\end{array}
\]
So
\[
\begin{array}{rcl}
|I_2|\bullet \alpha &=& T\bullet \alpha[a_0:= \alpha(x_0)][a_1:=\alpha(x_1)][a_2:=1][a_3:=1][a_3:=-1]\\
T\bullet \alpha &=&  \begin{cases}
\alpha[y:=\alpha (a_2)] & \text{if } \alpha(a_1)=0,\\
T\bullet \alpha[a_2:= \alpha (a_2)\cdot \alpha (x_0)][a_1:=\alpha(x_1)-1] & \text {otherwise.}
\end{cases}
\end{array}
\]
It follows that $|I_2|\bullet \alpha \neq \di$ if and only if $\alpha(x_1)= n$ for some $n\in \Nat$. Hence 
\[
\semone{I_2}_\mathcal{M}^{1}(m_0,m_1)=
\begin{cases}
m_0^{m_1} & \text{ if } m_1=n \text{ for some }n\in \Nat\\
\uparrow & \text{otherwise}
\end{cases}
\]
\end{enumerate}
for every meadow $\mathcal{M}$. In case $\mathcal{M}= \mathbb{Q}$, $I_2$ yields a non-total mapping; on 
prime fields
$\mathbb{Z}/p\mathbb{Z}$|considered zero-totalized| this mapping is total.
\end{example}
Every meadow mapping that is total on all meadows is clearly total on all cancellation meadows. The converse, however,
does not hold: consider the instruction sequence
\[
I=a_0.\cp (x_0); - a_0.\test ; \#2;\#4; a_0.\seta (a_0); + a_0.\test; \#0; y.\cp (a_1); !.
\]
Given any meadow $\mathcal{M}$, we have
\[
\semone{I}_\mathcal{M}^{0}(m)=
\begin{cases}
0 & \text{ if } m=0,\\
0 & \text{ if } m\neq 0\ \&\ 2m\neq 0,\\
\uparrow & \text{ otherwise}.
\end{cases}
\]
In the absence of proper zero divisors, $m=0$ if $2m=0$. Thus $\semone{I}^{0}_\mathcal{M}$ is the constant zero mapping on every 
cancellation meadow $\mathcal{M}$. On the zero-totalized field $\mathbb{Z}/6\mathbb{Z}$, however,  $3\neq 0$ and 
$2\times 3 = 0$, and thus $\semone{I}^{0}_{\mathbb{Z}/6\mathbb{Z}}(3)\uparrow$. 
\section{Characterization of total calculation on cancellation meadows}
In this section we shall prove the main theorem.

Total mappings share the following \emph{finite  representation property}.
\begin{proposition}
Let $T$ be a regular thread and $k\in \Nat$.  If $\semone{T}^{k}_\mathcal{M}$ is total on all
meadows $\mathcal{M}$, then there exists a finite thread $T'$ such that  
$\semone{T}^{k}_\mathcal{M}=\semone{T'}^{k}_\mathcal{M}$ for all meadows $\mathcal{M}$.
\end{proposition}
\begin{proof} Consider the set $\Gamma$ consisting of all meadow axioms together with the infinite
set $\{ \neg R_{T,k,n}(c_0, \ldots , c_k) \mid n\in \Nat\}$. If $\Gamma$ is finitely satisfiable, it must be 
simultaneously satisfiable, by Compactness, say in some meadow $\mathcal{M}$. This means that $\semone{T}_\mathcal{M}^k$ is not total, contradicting
 the assumption. We may therefore assume  that $\Gamma$ is not finitely satisfiable. By a standard model-theoretic argument and the monotonicity of the 
 termination predicate (Lemma \ref{pred} (3)),  it follows that for some $n\in \Nat$ and all meadows $\mathcal{M}$, $\mathcal{M}\models \forall x_0, \ldots, x_k\ 
 R_{T,k,n}(x_0, \ldots , x_k)$. Hence $\mathcal{M}\models \forall x_0, \ldots, x_k\ R_{\pi_{n+1}(T),k,n}(x_0, \ldots , x_k)$ for all $\mathcal{M}$ by Lemma \ref{pred} (4). 
 Then $\semone{T}_\mathcal{M}^k = \semone{\pi_{n+1}(T)}_\mathcal{M}^k$ for all meadows $\mathcal{M}$. \end{proof} 

\begin{proposition}\label{totalthread}
Let $T$ be a regular thread and $k\in \Nat$.  If $\semone{T}^{k}_\mathcal{M}$ is total on all
cancellation meadows $\mathcal{M}$, then there exists a finite thread $T'$ such that  
$\semone{T}^{k}_\mathcal{M}=\semone{T'}^{k}_\mathcal{M}$ for all cancellation meadows $\mathcal{M}$.
\end{proposition}
\begin{proof} Repeat the previous proof with $\Gamma$ supplemented with the cancellation axiom \\
\[\forall x,y,z\ (x\neq 0\ \&\ x\cdot y =x \cdot z \longrightarrow
y=z).\] \end{proof}

Next we shall show that tests can be abandoned without the loss of expressive power.

For $t\in \ter(\Sigma_m,\var)$, $\semone{t}_{\mathcal{M},\alpha}$ denotes the interpretation of $t$ in the 
meadow $\mathcal{M}$ under 
the assignment $\alpha$, and if $\sigma \in \ter(\Sigma_m,\var)^\var$, then $t^\sigma$ is the result of substituting all variables $v$ occurring in $t$ by 
$\sigma (v)$.
Recall that substitutions and assignments
interact in the following way.
\begin{lemma}\label{sublemma}
Let $\mathcal{M}$ be a meadow, $\alpha \in \mathcal{M}^\var$ an assignment and $\sigma\in \ter(\Sigma_m,\var)^\var$ a substitution.
Define $\alpha '\in \mathcal{M}^\var$ by $\alpha '(v)=\semone{v^\sigma}_{\mathcal{M},\alpha}$. Then for all $t\in \ter(\Sigma_m,\var)$,
\[
\semone{t}_{\mathcal{M}, \alpha '}= \semone{t^\sigma}_{\mathcal{M}, \alpha}.
\]
\end{lemma}

\begin{proposition}\label{fromT2t}
Let $T$ be a finite thread and $k\in \Nat$. Then there exists a term
$t_{T}\in \ter(\Sigma_m,\{x_0, \ldots , x_k\})$ such that for all cancellation meadows $\mathcal{M}$ and all 
$m_0, \ldots, m_k\in \mathcal{M}$, 
\[
\semone{T}^k_{\mathcal{M}}(m_0, \ldots , m_k)\downarrow \longrightarrow \semone{T}^k_{\mathcal{M}}(m_0, \ldots , m_k)=
\semone{t_{T}}_{\mathcal{M},\alpha_{m_0, \ldots , m_k}}.
\]
\end{proposition}
\begin{proof}
We use induction loading and 
employ structural induction on $T$ in order to prove the assertion stating the existence of a term $t_{T}\in 
\ter(\Sigma_m,\var)$ such that for all cancellation meadows $\mathcal{M}$ and all assignments $\alpha \in
\mathcal{M}^\var$, 
\[
(T\bullet \alpha )(y)=\semone{t_{T}}_{\mathcal{M},\alpha}
\]
if $T\bullet \alpha \neq \di$.

If $T=\st$, then 
\[\st\bullet \alpha)(y)=\alpha (y)= \semone{y}_{\mathcal{M},\alpha}.
\]
Hence $t_{\st}\equiv y$. If $T=\di$, we also put $t_{\di}\equiv y$. For the induction step, we have to distinguish 9 cases each of which corresponds to one
the 9 instructions sorts in $Ins(\mathbb{Q})$. The assignment instructions are proven straightforwardly using
the previous substitution lemma. We show 3 cases.

Suppose $T= a_i.\cp (x_j)\circ T'$ and $T\bullet \alpha \neq \di$. Then
\[
\begin{array}{rcll}
(T\bullet \alpha)(y)&=&(T'\bullet \alpha[a_i:=\alpha(x_j)])(y)\\
&=&\semone{t_{T'}}_{\mathcal{M},\alpha[a_i:=\alpha(x_j)]}&\text{by the induction hypothesis}\\
&=&\semone{t^\sigma_{T'}}_{\mathcal{M},\alpha}&\text{by Lemma \ref{sublemma}}
\end{array}
\]
where $\sigma(a_i) = x_j$, and $\sigma(v) = v$ if $v\not\equiv  a_i$. Hence $t_T\equiv t^\sigma_{T'}$ suffices. Likewise, if $T= a_i.\seta (a_j) \circ T'$ and $T\bullet \alpha \neq 
\di$, then
\[
\begin{array}{rcll}
(T\bullet \alpha)(y)&=&(T'\bullet \alpha[a_i:=\alpha (a_i) + \alpha (a_j)])(y)\\
&=&\semone{t_{T'}}_{\mathcal{M},\alpha[a_i:=\alpha (a_i) + \alpha (a_j)]}&\text{by the induction hypothesis}\\
&=&\semone{t_{T'}^\sigma}_{\mathcal{M},\alpha}&\text{by Lemma \ref{sublemma}}
\end{array}
\]
where $\sigma(a_i) = a_i+a_j$, and $\sigma(v)=v$ if $v\not \equiv a_i$. And if $T= y.\cp (a_j) \circ T'$ and $T\bullet \alpha \neq \di$, then
\[
\begin{array}{rcll}
(T\bullet \alpha)(y)&=&(T'\bullet \alpha[y:=\alpha (a_j)])(y)\\
&=&\semone{t_{T'}}_{\mathcal{M},\alpha[y:=\alpha (a_j)]}&\text{by the induction hypothesis}\\
&=&\semone{t_{T'}^\sigma}_{\mathcal{M},\alpha}&\text{by Lemma \ref{sublemma}}
\end{array}
\]
where $\sigma(y)=a_j$, and $\sigma(v)=v$ if $v\not \equiv  y$.  

The case for the zero test
 exploits the fact that in every cancellation meadow $\mathcal{M}$ we have
\[
\semone{0_{a_i}\cdot t + 1_{a_i} \cdot t'}_{\mathcal{M},\alpha}=
\begin{cases}
\semone{t}_{\mathcal{M},\alpha}& \text{ if  } \alpha(a_i) = 0\\
\semone{t'}_{\mathcal{M},\alpha}& \text{ otherwise  } 
\end{cases}
\]
Hence, if $T= T'\trianglelefteq a_i.\test   
\trianglerighteq T''$ we can take 
\[
t_{T}\equiv 0_{a_i}\cdot t_{T'} + 1_{a_i }
\cdot t_{T''}. 
\]

The original assertion now follows from the observation that we can replace all occurrences of auxiliary variables, 
the output variable $y$ and all input variables $x_n$ with $k<n$ in the term $t_T$ by $0$ if $\alpha$ is initial. \end{proof}
\begin{definition}
We shall say that the thread $T$ \emph{computes} $t\in \ter(\Sigma_m,\{x_0, \ldots , x_k\})$, if for all  cancellation meadows 
$\mathcal{M}$ and all
$m_0, \ldots , m_k \in \mathcal{M}$, 
\[
\semone{T}^k_{\mathcal{M}}(m_0, \ldots , m_k)=\semone{t}_{\mathcal{M},\alpha_{m_0, \ldots , m_k}}.
\]
\end{definition}
Thus if $T$ is finite, the free variables of $t_T$ are among $\{x_0, \ldots , x_k\}$ and $\semone{T}^k_{\mathcal{M}}$ is total,
then $T$ computes the term $t_T$. Conversely, every meadow term $t$ with free variables 
in $\var_{in}$ can be computed by a finite thread $T_t$ which is in addition 
\emph{test-free}|that is, postconditional composition occurs as action prefix 
only|and which uses at most 5 auxiliary  variables. To these ends, we shall define the \emph{raise $T^1$ of a thread $T$}  as the thread $T'$ obtained
from $T$ by raising the subscript of every auxiliary variable occurring in $T$ by 1.
\begin{definition}
\begin{enumerate}
\item Let $i\in \ins (\mathbb{Q})$ be an instruction. Then $i^1$ is defined by
\[
\begin{array}{rcll}
a_i.\cp (x_j)^1&=& a_{i+1}.\cp (x_j)\\
a_i.\setz^1 &=&a_{i+1}.\setz  \\
a_i.\seto^1&=&a_{i+1}.\seto \\
a_i.\setai^1&=& a_{i+1}.\setai  \\
a_i.\setmi^1&=& a_{i+1}.\setmi \\ 
a_i.\seta (a_j)^1 &=& a_{i+1}.\seta (a_{j+1}) \\
a_i.\setm (a_j)^1&=&a_{i+1}.\setm (a_{j+1}) \\
a_i.\test ^1&=& a_{i+1}.\test \\
y.\cp (a_j)^1&=& y.\cp (a_{j+1})
\end{array}
\]
\item Let $T$ be a finite thread. Then $T^1$ is defined inductively by
$\st^1=\st$, $\di^1=\di$, and $(T'\trianglelefteq i  \trianglerighteq T'')^1=T'^1\trianglelefteq i^1  \trianglerighteq T''^1$ 
for $i\in \ins (\mathbb{Q})$.
\end{enumerate}
\end{definition} 
A thread and its raise compute the same values.
\begin{lemma}
Let $T$ be a finite thread and $k\in \Nat$.  Then for all  meadows $\cal{M}$ and all $m_0, \ldots , m_k\in \cal{M}$
\[
\semone{T}^k_{\mathcal{M}}(m_0, \ldots , m_k)\downarrow \longrightarrow \semone{T}^k_{\mathcal{M}}(m_0, \ldots , m_k)=
\semone{T^1}^k_{\mathcal{M}}(m_0, \ldots , m_k).
\]
\end{lemma}
\begin{proof}
For the sake of the proof, we define for $\alpha \in \mathcal{M}^\var$ the raise
$\alpha^1\in \mathcal{M}^\var$ by 
\[
\alpha^1(v)=
\begin{cases}
\alpha (a_{i+1}) & \text{ if } v \equiv a_i,\\
\alpha (v )& \text{ if } v\not\in \var_{aux}.
\end{cases}
\]
We now show that 
\[
T\bullet \alpha^1\neq \di \longrightarrow (T\bullet \alpha^1)(y)=(T^1\bullet \alpha)(y)
\]
by structural induction on $T$. If $T=\st$, then
\[ (T\bullet \alpha^1)(y)= (\st\bullet \alpha^1)(y)= \alpha^1 (y)= \alpha (y)=( \st\bullet \alpha )(y) = 
(T^1\bullet \alpha )(y).\]
For the induction step, we have to distinguish 9 cases each of which corresponds to one of the 9 instruction sorts.
Each case follows straightforwardly. We show the case that $T=a_i.\cp (x_j) \circ T'$:
\[
\begin{array}{rcl}
(T\bullet \alpha^1)(y)&=&((a_i.\cp (x_j) \circ T')\bullet \alpha^1)(y)\\
&=&(T'\bullet \alpha^1[a_i:=\alpha^1(x_j)])(y)\\
&=&(T'\bullet (\alpha [a_{i+1}:=\alpha (x_j)])^1)(y)\\
&=&(T'^1\bullet \alpha [a_{i+1}:=\alpha (x_j)])(y)\text{ by the induction hypothesis}\\
&=&(a_{i+1}.\cp (x_j)\circ T'^1\bullet \alpha )(y)\\
&=&(T^1\bullet \alpha )(y)\\
\end{array}
\] 
The statement now follows from the the observation that $\alpha^1 = \alpha$ if $\alpha$ is initial.
\end{proof}
In the sequel, we shall say that a thread $T$ (an instruction sequence $I$)
\emph{uses
the auxiliary variable $a_i$}, if the variable $a_i$ occurs in at least one of $T$'s ($I$'s) atomic actions (basic instructions). Moreover, we shall say that
$T$ ($I$) \emph{uses  $n$ auxiliary variables}, if $T$ ($I$) uses precisely the auxiliary variables 
$a_0, \ldots , a_{n-1}$.
\begin{lemma}\label{auxnumber}
\begin{enumerate}
\item 
If $t\in \var_{in}\cup \{0,1\}$, then $t$ can be computed by a finite and test-free thread that uses 1 auxiliary
variable.
\item 
If $t$ can be computed by a finite and test-free thread that uses $n$ auxiliary
variables, then so can $-t$ and $t^{-1}$.
\item Suppose $t,t'\in \ter (\Sigma_m, \var_{in})$
can be computed by finite and test-free threads that use $n$ and $m$ auxiliary variables, 
respectively. If $n=m$, then $t+t'$ and $t\cdot t'$ can be computed by  finite and test-free threads that
use $n+1$ auxiliary variables, and if $n\neq m$, then $t+t'$ and $t\cdot t'$ can be 
computed finite and test-free threads that
use $\max\{n,m\}$ auxiliary variables.
\end{enumerate}
\end{lemma}
\begin{proof}
We shall construct appropriate threads of the form
\[
i_1 \circ \cdots \circ i_k \circ y.\cp (a_0) \circ \st
\]
\begin{enumerate}
\item Observe that
$
a_0.\cp(x_i)\circ y.\cp (a_0) \circ \st
$, 
$
a_0.\setz\circ y.\cp (a_0) \circ \st
$, and 
$
a_0.\seto\circ y.\cp (a_0) \circ \st
$
compute $x_i$, $0$ and $1$, respectively.
\item Suppose $T=i_1 \circ \cdots \circ i_k \circ y.\cp (a_0) \circ \st$
computes $t$.
Then
\[
i_1 \circ \cdots \circ i_k \circ a_0.\setai \circ y.\cp (a_0) \circ \st
\]
computes $-t$ and 
\[
i_1 \circ \cdots \circ i_k \circ a_0.\setmi \circ y.\cp (a_0) \circ \st
\]
computes $t^{-1}$. Both threads use as many auxiliary variables as $T$ and are finite and test-free.
\item Suppose $T=i_1 \circ \cdots \circ i_k \circ y.\cp (a_0) \circ \st$ uses $n$ auxialiary 
variables to compute $t$, and $T'=j_1 \circ \cdots \circ j_l \circ y.\cp (a_0) \circ \st$
uses $m$ auxiliary variables to compute $t'$. We first assume that $n\leq m$. Since
by the previous lemma $T^1$ also computes $t$, we have that 
\[
j_1 \circ \cdots \circ j_l \circ a_1.\setz \circ \cdots \circ a_{n}.\setz \circ 
i_1^1 \circ \cdots \circ i_k^1 \circ a_0.\seta(a_1)\circ y.\cp (a_0) \circ \st
\] 
computes $t+t'$. This thread is finite and test-free, and uses $n+1$ auxiliary variables if $n=m$ and otherwise
$m$ variables. If $m<n$ then
\[
i_1 \circ \cdots \circ j_k \circ a_1.\setz \circ \cdots \circ a_{m}.\setz \circ 
j_1^1 \circ \cdots \circ j_l^1 \circ a_0.\seta(a_1)\circ y.\cp (a_0) \circ \st
\] 
uses $n$ auxiliary variables and computes $t'+t$ and hence $t+t'$.

For $t\cdot t'$ we replace the action $a_0.\seta(a_1)$ in the above threads by $a_0.\setm(a_1)$.
\end{enumerate}
\end{proof}

\begin{proposition}\label{5aux}
Let $t\in \ter (\Sigma_m, \var_{in})$ be a meadow term. Then $t$ can be computed by a finite, test-free thread $T$ that uses 
5 auxiliary variables.
\end{proposition}
\begin{proof}
By Corollary~\ref{smf} $t$ can be represented as a sum of quotients of polynomials. Moreover, every polynomial
can be written as a sum of monomials, i.e.\ expressions of the form 
$n\cdot x_{i_1}\cdots x_{i_k}$ or $-n\cdot x_{i_1}\cdots x_{i_k}$. 
Since $n=1+ \cdots +1+0 +0$ it can be computed by a finite and test-free thread that uses 2 auxiliary variables 
by~\ref{auxnumber}.(1) and (3). Thus also $n\cdot x_{i_1}, \ldots , n\cdot x_{i_1}\cdots x_{i_k}$ can all be computed
by finite and test-free threads that use 2 auxiliary variables. And the same holds for 
$-n\cdot x_{i_1} \cdots x_{i_k}$ by \ref{auxnumber}.(2). Thus every monomial can be computed
by a finite and test-free thread that uses  2 auxiliary variables. It follows that every sum of monomials|and hence 
every polynomial|can be computed
by a finite and test-free thread that uses 3 auxiliary variables by \ref{auxnumber}.(3). Whence every quotient of polynomials
can be computed by a finite and test-free thread that uses 4 auxiliary variables by \ref{auxnumber}.(2) and (3). Invoking 
again \ref{auxnumber}.(3) we obtain that
every sum of quotients of polynomials|and therefore $t$|can be computed by a finite and test-free 
thread that uses 
5 auxiliary variables.
\end{proof}
Summarizing we have proven the following completeness result.
\begin{theorem}
Let $I$ be an instruction sequence and $k\in \Nat$ be such that $\semone{I}_\mathcal{M}^k$ is a total mapping on all 
cancellation meadows $\mathcal{M}$. Then there exists a  straight-line instruction sequence $J$
which uses at most 5 auxiliary variables such that
 $\semone{I}_\mathcal{M}^k=\semone{J}_\mathcal{M}^k$ for all cancellation meadows $\mathcal{M}$. 
\end{theorem}
\begin{proof}
Suppose that $\semone{I}_\mathcal{M}^k$ is total on all cancellation meadows $\mathcal{M}$. By Proposition~\ref{totalthread}, we can pick a finite thread $T$ such that $\semone{I}_\mathcal{M}^k=
\semone{T}_\mathcal{M}^k$ for all cancellation meadows $\mathcal{M}$. By Proposition~\ref{fromT2t} we may assume that 
$T$ computes the term $t\in \ter(\Sigma_m,\{x_0, \ldots , x_k\})$ which in turn is computed by a 
finite and test-free thread $T'$ that uses 5 auxiliary variables by the previous proposition. We can now take a straight-line instruction sequence $J$ with $|J|=T'$.
\end{proof}
\section{Calculation on signed cancellation meadows}
We obtain \emph{signed meadows} by 
extending the signature $\Sigma_m$ of meadows with the unary sign function $\sg(\_)$.
We write $\Sigma_{ms}$ for this extended signature, so 
$\Sigma_{ms} = (0,1,+,\cdot,-,^{-1},\sg)$.
The sign function $\sg$
presupposes an ordering $<$ of its domain and is defined as follows:

\[\sg(x)=\begin{cases}
-1&\text{if }x<0,\\
0&\text{if }x=0,\\
1&\text{if }x>0.
\end{cases}\]

One can define $\sg$ in an equational manner
by the set \SA\ of axioms 
given in Table~\ref{t:sign}. 
\begin{table}[hbtp]
\centering
\rule[-2mm]{8.6cm}{.5pt}
\begin{align}
\label{ax1}
\sg(1_x)&=1_x\\
\label{ax2}
\sg(0_x)&=0_x\\
\label{ax3}
\sg(-1)&=-1\\
\label{ax4}
\sg(x^{-1})&=\sg(x)\\
\label{ax5}
\sg(x\cdot y)&=\sg(x)\cdot \sg(y)\\
\label{ax6}
0_{\sg(x)-\sg(y)}\cdot (\sg(x+y)-\sg(x))&=0
\end{align}
\rule[3mm]{8.6cm}{.5pt}
\vspace{-5mm}
\caption{The set \SA\ of axioms for the sign function}
\label{t:sign}
\end{table}
\noindent
First, notice that by \Md\ and axiom~\eqref{ax1} 
(or axiom~\eqref{ax2}) we find 
\[\sg(0)=0\quad\text{and}\quad\sg(1)=1.\]
Then, observe that in combination with the inverse law \IL, 
axiom \eqref{ax6} is an equational representation of the conditional
equational axiom
\[\sg(x)=\sg(y)~\longrightarrow ~\sg(x+y)=\sg(x).\]
The initial algebra of $\Md\cup\SA$ is $\mathbb{Q}$ expanded with the sign function. A proof follows immediately
from the techniques used in ~\cite{BT07,BR08}.

Some consequences of the $\Md\cup\SA$ are:
\begin{align}
\label{1}
\sg(x^2)&=1_x\text{ because 
$\sg(x^2)=\sg(x)\cdot\sg(x)
=\sg(x)\cdot\sg(x^{-1})=\sg(1_x)=1_x$,}\\
\label{2}
\sg(x^3)&=\sg(x)\text{ because 
$\sg(x^3)=\sg(x)\cdot\sg(x)\cdot\sg(x^{-1})
=\sg(x\cdot (x\cdot x^{-1}))=\sg(x)$,}\\
\label{3}
1_x\cdot\sg(x)&=\sg(x)\text{
because $1_x\cdot\sg(x)=\sg(x^2)\cdot\sg(x)
=\sg(x^3)=\sg(x)$,}\\
\label{5}
\sg(x)^{-1}&=\sg(x)\text{
because $\sg(x)^{-1}\begin{array}[t]{l}=(\sg(x)^2\cdot\sg(x)^{-1})^{-1}=(\sg(x^2)\cdot\sg(x)^{-1})^{-1}
\\=(1_x\cdot\sg(x)^{-1})^{-1}=1_x\cdot \sg(x)=\sg(x).\end{array}$}
\end{align}

So, $0=\sg(x)-\sg(x)=\sg(x)-\sg(x)^3=\sg(x)(1-\sg(x)^2)$ and hence
\begin{equation}
\label{4}
\sg(x)\cdot(1-\sg(x))\cdot (1+\sg(x))=0.
\end{equation}

The \emph{finite basis result}
for the equational theory of cancellation meadows is formulated in a generic way so that
it can be used for any expansion of a meadow that satisfies
the propagation properties defined below.

\begin{definition}
Let $\Sigma$ be an extension of $\Sigma_m=(0,1,+,\cdot,-,^{-1})$, the
signature of meadows. Let $E\supseteq \Md$ (with \Md\ the set
of axioms for meadows given in Table~\ref{Md}).
\begin{enumerate}
\item
$(\Sigma,E)$ has the \textbf{propagation property for pseudo units} if for
each pair of $\Sigma$-terms $t,r$ and context $C[~]$,
\[E\vdash 1_t\cdot C[r]=1_t\cdot C[1_t\cdot r].\]
\item
$(\Sigma,E)$ has the \textbf{propagation property for pseudo zeros} if for
each pair of $\Sigma$-terms $t,r$ and context $C[~]$,
\[E\vdash 0_t\cdot C[r]=0_t\cdot C[0_t\cdot r].\]
\end{enumerate}
\end{definition}
Preservation of these propagation properties admits the following 
nice result:
\begin{theorem}
[Generic Basis Theorem for Cancellation Meadows]
\label{st}
If $\Sigma\supseteq \Sigma_m,~ E\supseteq \Md$ and $(\Sigma,E)$ 
has the
pseudo unit and the pseudo zero propagation property,
then $E$ is a basis (a complete axiomatisation) of 
$\Mod_\Sigma(E\cup\IL)$.
\end{theorem}
Bergstra and Ponse~\cite{BP08} proved that $\Md$ and $\Md\ {\cup} \SA$ satisfy both propagation properties and
are therefore complete axiomatizations of 
$\Mod_\Sigma(\Md\cup\IL)$ and $\Mod_\Sigma(\Md \cup \SA\cup\IL)$, respectively. Since
\[
\Md \cup \SA\cup\IL\vdash t=0_{t'}\cdot s + 1_{t'} \cdot s' \longrightarrow \sg(t)=0_{t'}\cdot \sg(s) + 1_{t'} \cdot \sg(s')
\]
using $\IL$ and the axioms (1), (2) and (5) of $\SA$, it then follows that
\[
\Md \cup \SA\vdash t=0_{t'}\cdot s + 1_{t'} \cdot s' \Longrightarrow  \Md \cup \SA\vdash\sg(t)=0_{t'}\cdot \sg(s) + 1_{t'} \cdot \sg(s').\  \  \  (\dag)
\]
We can hence adapt the Standard Meadow Form to signed meadow terms as follows.

We write $\Sigma_{rs}=(0,1,+,\cdot,-,\sg)$ for the
signature of signed rings.
A \emph{signed polynomial} is then an expression over $\Sigma_{rs}$, thus  without 
inverse operator. 

\begin{definition}
A term $t\in \ter(\Sigma_{ms},X)$ is a \emph{Standard Signed Meadow Form
(SSMF)} if, for some $n\in\Nat$, $t$ is an \emph{SSMF of level $n$}.
SSMFs of level $n$ are defined as follows:
\begin{enumerate}
\item \textit{SSMF of level $0:$} each expression of the form $s/t$ 
with $s$ and $t$ ranging over signed polynomials,
\item \textit{SSMF of level $n+1:$} each expression of the form
\[0_{t'}\cdot s+1_{t'}\cdot t
\]
with $t'$ ranging over {signed polynomials} and $s$ and $t$ over SSMFs of 
level $n$.
\end{enumerate}
\end{definition}
\begin{theorem}
\label{SSMF}
For each  $t\in \ter(\Sigma_{ms},X)$ there exist  an SSMF $t_{SSMF}$
with the same variables such that $\Md \cup \SA \vdash t=t_{SSMF}$.
\end{theorem}
\begin{proof}
As in \cite{BP08} using $(\dag)$.
\end{proof}

As in Corollary~\ref{smf} it follows that every signed meadow term is provably equal to a sum of quotients of signed polynomials.
\begin{corollary}\label{ssmf}
For every  $t\in \ter(\Sigma_{ms},X)$ there exist signed polynomials $s_0, t_0, \ldots , s_n, t_n$ such that
\[
\Md \cup \SA \vdash t=\frac{s_0}{t_0} + \ldots + \frac{s_n}{t_n}. 
\]
\end{corollary}

Signed polynomials also enjoy a standard form.
\begin{lemma}\label{signednf}
Let $t$ be a signed polynomial and $n\in \Nat$ be the number of its subterms of the form $\sg(t')$. Then there are polynomials 
$t_1, t_{1_1}, \ldots , t_{1_n}, \ldots , t_i, t_{i_1}, \ldots , t_{i_n}, \ldots , t_{3n}, t_{3n_1}, \ldots , t_{3n_n}$ such that
 \[
\Md \cup \SA \vdash t=\Sigma_{i=1}^{3^n} \Pi_{j=1}^n 0_{\phi(\sg(t_{i_j}))}\cdot t_i 
\]
where $\phi(\sg(t_{i_j}))\in\{\sg(t_{i_j}), 1+ \sg(t_{i_j}), 1- \sg(t_{i_j})\}$.
\end{lemma} 
\begin{proof}
We employ induction on the number $n$ of subterms of the form $\sg(t')$. If $n=0$ then $t$ itself is a polynomial and hence $t_1\equiv t$ suffices. 

Suppose $n=l+1$ and pick an innermost subterm $\sg (t')$ of $t$. Then $t\equiv C[\sg (t')]$ for some context $C$ and polynomial $t'$. From $\IL$ together with (11) it follows that $\sg(t')=0$ or $\sg(t')=1$ or $\sg(t')=-1$. Thus
\[
Md \cup \SA \vdash t=0_{\sg(t')}\cdot C[0] + 0_{1-\sg(t')}\cdot C[1] + 0_{1 +\sg(t')}\cdot C[-1]
\]
with $C[0], C[1],$ and $C[-1]$ having $l$ signed subterms. We can now apply the induction hypothesis.
\end{proof}

A suitable instruction for computations on signed meadows is $a.\sets$ with Boolean reply $\tr$ and
the obvious semantics
$a \Leftarrow \sg(a)$. We add this instruction to $Ins(\mathbb{Q})$, and consider instruction sequences and corresponding threads over the enriched instruction set in the sequel.
\begin{example}
Notice that, with the sign function available, the function $\max(x_0,x_1)$ has the following simple definition
\[
\max(x_0,x_1)=\begin{cases}
(\sg(x_0) + 1)\cdot x_0/2 & \text{ if $x_1=0$}\\
\max(x_0-x_1,0) + x_1 & \text{ otherwise}.
\end{cases}
\]
$\max(x,y)$ can be computed by the periodic  instruction sequence
\[
\begin{array}{l}
a_0.\cp(x_0); a_1.\cp(x_1); a_2.\setz; a_4.\cp(x_0); a_1.\test; \#2; \#11;\\[1mm]
(a_3.\seto; a_4.\sets; a_4.\seta(a_3); a_0.\setm(a_4); \\[1mm]
\ a_3.\seta(a_3);  a_3.\setmi; a_0.\setm(a_3); a_0.\seta(a_2); y_0.\cp(a_0); !;\\[1mm]
\ a_1.\setai; a_0.\seta(a_1); a_4.\seta(a_1); a_2.\cp(x_1))^\omega
\end{array}
\]
which also has a finite representation.
\end{example}
The termination predicate and the apply operator can both be extended to regular threads using
the sign instruction in the obvious way by 
\[
\begin{array}{rcl}
R_{a_i.\sets \circ T, n+1}^\mathcal{M} &=  &\{\alpha \in \mathcal{M}^{\var}\mid\alpha 
[a_i:=\sg (\alpha(a_i))]\in R_{T,n}^\mathcal{M}\}\text{, and}\\[1mm]
(a_i.\sets \circ T) \bullet \alpha &=& T \bullet \alpha [a_i:=\sg(\alpha(a_i))].
\end{array}
\]
We then have the following completeness result.
\begin{theorem}
Let $I$ be an instruction sequence and $k\in \Nat$ be such that $\semone{I}_\mathcal{M}^k$ is a total mapping on all 
signed cancellation meadows $\mathcal{M}$. Then there exists a  straight-line instruction sequence $J$
which uses at most 8 auxiliary variables such that
 $\semone{I}_\mathcal{M}^k=\semone{J}_\mathcal{M}^k$ for all cancellation meadows $\mathcal{M}$. 
\end{theorem}
\begin{proof}
The propositions~\ref{totalthread} and~\ref{fromT2t} extend straightforwardly to signed cancellation
meadows. Thus $|I|$ computes a term 
$t\in \ter(\Sigma_{ms},\{x_0, \ldots x_k\})$. 
It remains to show that $t$ can be computed by a finite and test-free thread that uses at most 8 auxiliary variables.
From Corollary~\ref{ssmf} it follows that $t$ is provably equal to a sum of quotients of signed polynomials. Then,
following the proof of Proposition~\ref{5aux}, it suffices to prove that a signed polynomial can be computed by a finite and test-free
thread using at most 6 auxiliary variables. To these ends, we invoke Lemma~\ref{signednf}. Thus we may assume that there exist polynomials
$t_1, t_{1_1}, \ldots , t_{1_n}, \ldots , t_i, t_{i_1}, \ldots , t_{i_n}, \ldots , t_{3n}, t_{3n_1}, \ldots , t_{3n_n}$ such that 
\[
t=\Sigma_{i=1}^{3^n} \Pi_{j=1}^n  0_{\phi(\sg(t_{i_j}))}\cdot t_i
\]
where $\phi(\sg(t_{i_j}))\in\{\sg(t_{i_j}), 1+ \sg(t_{i_j}), 1- \sg(t_{i_j})\}$. 

From Lemma~\ref{auxnumber} it follows that a polynomial $t'$ can be computed by a finite and test-free 
thread using 3 auxiliary variables. Say 
\[
i_1 \circ \cdots \circ i_k \circ y.\cp (a_0) \circ \st
\]
computes $t'$. Then 
\[
i_1 \circ \cdots \circ i_k \circ a_0.\sets(a_0)\circ y.\cp (a_0) \circ \st
\]
computes $\sg (t')$ using the same variables.
Thus also
$\phi(\sg(t'))$ can be computed by a finite and test-free thread 
using 3 auxiliary variables. Hence $0_{\phi(\sg(t_{i_j}))}\cdot t_i$ can be computed with 
4 auxiliary variables by a finite and test-free thread. Therefore it takes at most 5 auxiliary variables to compute 
$\Pi_{j=1}^n  0_{\phi(\sg(t_{i_j}))}\cdot t_i$ and  6 to compute $t$ by a finite thread without any tests.
\end{proof}

\section{Conclusions and future work}
We have described  an algebraic execution system that can be used to analyze properties of instruction sequences. It is especially
designed to perform calculation on the signed rational numbers. We have proven that total instruction sequences can be computed by straight-line
programs with a bound supply of auxiliary variables.

Important computer algorithms based on discrete Fourier transformations can be expressed within the signed rational numbers extended with
$\sin$ and $\pi$. For future work, we aim at examining equivalence and simplification problems for this kind of straight-line 
instruction sequences. However, it 
is yet unclear to us where straightening starts to fail.

\end{document}